%% file: main.tex
\documentclass[UKenglish,cleveref]{lipics-v2021}
\usepackage{main}
\usepackage[color=cherry-blossom]{todonotes}
\nolinenumbers

\title{Encoding of Predicate Subtyping with Proof Irrelevance
  in the \(\mathsf{\lambda\Pi}\)-calculus Modulo Theory}
\titlerunning{Predicate Subtyping with Proof Irrelevance in LPMT}
\author{Gabriel Hondet}{%
  Université Paris-Saclay, ENS Paris-Saclay, CNRS, Inria,
  Laboratoire Méthodes Formelles, 91190, Gif-sur-Yvette, France
}{gabriel.hondet@inria.fr}{}{}
\author{Frédéric Blanqui}{%
  Université Paris-Saclay, ENS Paris-Saclay, CNRS, Inria,
  Laboratoire Méthodes Formelles, 91190, Gif-sur-Yvette, France
}{frederic.blanqui@inria.fr}{}{}
\authorrunning{G. Hondet and F. Blanqui}
\Copyright{Inria}



\EventEditors{U. de'Liguoro and S. Berardi and T. Altenkirch}
\EventNoEds{3}
\EventLongTitle{Post-proceedings of the 26th International Conference on Types for Proofs and Programs}
\EventShortTitle{TYPES 2020}
\EventAcronym{TYPES}
\EventYear{2020}
\EventDate{March 2--5, 2020}
\EventLocation{University of Turin, Italy}
\EventLogo{}
\SeriesVolume{188}
\ArticleNo{6}

\begin{document}
\maketitle

\begin{abstract}
  The $\lambda\Pi$-calculus modulo theory is a logical framework in
  which various logics and type systems can be encoded, thus helping
  the cross-verification and interoperability of proof systems based
  on those logics and type systems. In this paper, we show how to
  encode \emph{predicate subtyping} and \emph{proof irrelevance}, two
  important features of the \pvs proof assistant. We prove that this
  encoding is correct and that encoded proofs can be mechanically
  checked by \Dedukti, a type checker for the
  \(\mathrm{\lambda\Pi}\)-calculus modulo theory using rewriting.
\end{abstract}

\input{introduction}
\input{description}
\input{encoding}

\input{encoding_properties}
\input{eq_to_rewriting}
\input{conclusion}

\bibliography{translation}
\end{document}

%% file: introduction.tex
\section{Introduction}

A substantial number of proof assistants can be used to develop formal proofs,
but a proof developed in an assistant cannot, in general,
be used in another one.
This impermeability generates redundancy since theorems are likely to have one
proof per proof assistant.
It also prevents adoption of formal methods by industry because of the lack
of standards and the difficulty to use adequately formal methods.

Logical frameworks are a part of the answer.
Because of their expressiveness, different logics and proof systems can be
stated in a common language.
The \(\mathrm{\lambda\Pi}\)-calculus modulo theory, or \lpme,
is such a logical framework.
It is the simplest extension of simply typed \(\lambda\)-calculus
with dependent types and arbitrary computation
rules.
Fixed-length vectors are a common example of dependent type,
that can be represented in the \(\mathrm{\lambda\Pi}\)-calculus as \(\forall n: \N, \operatorname{Vec}(n)\).
The \(\mathrm{\lambda\Pi}\)-calculus modulo theory
already allows to formulate first order logic,
higher order logic~\cite{Assaf15HOL} or proof systems based
on \emph{Pure Type Systems} \cite{cousineau07tlca} such as \toolstyle{Matita} \cite{assaf15phd},
\toolstyle{Coq} \cite{boespflug12pxtp-coqine} or \toolstyle{Agda} \cite{genestier20fscd}.

\pvs~\cite{PVS:semantics} is a proof assistant
that has successfully been used in collaboration by academics and industrials
to formalise and specify real world systems~\cite{Owre98pvsExperienceReport}.
More precisely, \pvs is an environment comprising a specification language,
a type checker and a theorem prover.
One of the specificities of \pvs is its ability to blend type checking with
theorem proving by requiring terms to validate arbitrary predicates
in order to be attributed a certain type.
This ability is a consequence of
\emph{predicate subtyping}~\cite{OSR98predicatesubtyping}.
It facilitates the development of specifications and provides
a more expressive type system which allows to encode more constraints.
For instance, one can define the inverse function
\(\operatorname{inv}: \R^{*} \to \R\), where \(\R^{*}\) is a predicate subtype
defined as reals which are not zero.

If predicate subtyping provides a richer type system,
it also makes type checking of specifications undecidable.
In~\cite{gilbert2018phd}, F.~Gilbert paved the way of the expression
of \pvs into \lpme: he formalised the core of \pvs
and provided a language of certificates for \pvs whose type checking
is decidable.
However, the encoding in \lpme of this language of certificates relies on
\emph{proof irrelevance}.

The following work proposes an encoding of proof irrelevant
equivalences into the \(\lambda\Pi\)-calculus modulo theory.  It also
inspects the completion of such equations into a confluent rewrite
system.  The resulting rewrite system can be used to provide an
encoding of \pvs into \Dedukti, a type-checker for the
\(\lambda\Pi\)-calculus modulo theory based on rewriting
\cite{assaf19draft}.

\paragraph*{Related work}

An encoding or ``simulation'' of predicate subtyping \textit{\`a la} \pvs
into \toolstyle{HOL} can be found in~\cite{Hurd01pshol}.
The objective of that work was to get some facilities provided by
predicate subtyping into \toolstyle{HOL} rather than providing
a language of certificates, and proof checking hence remains undecidable.
Moreover, predicate subtypes are not represented by types but by theorems.

In~\cite{Sozeau06subsetcoercion},
predicate subtyping is weakened into a language named \toolstyle{Russell}
to be then converted into \textsc{CIC}.
This conversion amounts to the insertion of coercions
and unsolved meta-variables, the latter embody \pvs \emph{type correctness conditions} (\tcc).
The equational theory used in the \textsc{CIC} encoding is richer
than ours since it includes surjective pairing
\(e = \teZpair{T}{U}{(\teZprol{T}{U}{e})}{(\teZpror{T}{U}{e})}\)
and \(\eta\)-equivalence
\(f = \teZabst{x}{f~x}\) in addition to proof irrelevance.

In~\cite{Werner08proofIrrelevance}, proof irrelevance
is embedded into Luo's \textsc{ECC}~\cite{Luo90ECC} and its dependent pairs.
Pairs and dependent pair types come in two flavours,
the proof irrelevant one and the normal one.
The flavour is noted by an annotation, and proof irrelevance is implemented by
a reduction which applies only on annotated pairs.
The article presents as well an application to \toolstyle{pvs}.

On a slightly more practical side,
the automated first-order prover \toolstyle{ACL2}~\cite{KaufmannM97ACL2}
reproduces the system of ``guards'' provided by predicate subtyping
into its logic based on \toolstyle{Common Lisp}
with the concept of \emph{gold symbols}.
Approximately, a symbol is gold if all its \tcc have been solved.

Some theories---often based on Martin-L\"of's Type Theory---blend
together a decidable (called \emph{definitional} or \emph{intensional}) equality
with an undecidable (said \emph{extensional}) equality.
In~\cite{Pfenning01iepi}, a judgement ``\(A\) is provable'' is introduced,
to say that a proof of \(A\) exists, but no attention is paid to what it is.
Similarly, \cite{DBLP:journals/corr/abs-1102-2405} introduces proof irrelevance
in Martin-L\"of's logical framework using a function to distinguish propositions
\(A\) from ``proof-irrelevant propositions \(\mathsf{Prf}(A)\)''.
While \(A\) can be inhabited by several normal terms,
\(\mathsf{Prf}(A)\) is inhabited by only one normal form noted \(\star\),
to which all terms of \(\mathsf{Prf}(A)\) reduce.
Still in Martin-L\"of's type theory,
\cite{Salvesen88thestrength} provides proof irrelevance
for predicate subtyping (here called \emph{subset types})
for two different presentations, one is intensional, and the other extensional.
The interested reader may have a look at
\toolstyle{Nuprl}~\cite{Constable86implementingmathematics},
an implementation of Martin-L\"of's Type Theory
with extensional equality and subset types.

Proof irrelevance has also been added to LF to provide a new system LFI
in~\cite{pfenning10refinement}, where proof irrelevance is used
in the context of \emph{refinement types}.
In LFI, proof irrelevance is not limited to propositions,
nor it is attached to a certain type:
terms are irrelevant based on the function they are applied to.
A similar system is implemented in \toolstyle{Agda}~\cite{agdaManual}.

More generally, concerning proof irrelevance in proof assistants,
\toolstyle{Coq} and \toolstyle{Agda}~\cite{2019proofIrrWithoutK}
each have a sort for proof irrelevant propositions
(\texttt{SProp} for \toolstyle{Coq} and \texttt{Prop} for
\toolstyle{Agda}~\cite{agdaManual}).
\toolstyle{Lean}~\cite{leansysdesc} is by design proof irrelevant,
and \toolstyle{Matita} supports proof irrelevance as
well~\cite[section 9.3]{matitatutorial}.

\paragraph*{Outline}

Encoding predicate subtyping requires a clear definition of it, which is done in
\cref{sec:pcert}.
Predicate subtyping is encoded into \lpme using the signatures provided in
\cref{sec:encoding}. This encoding is put in use into some examples as well.
The encoding is proved correct in \cref{sec:translation}:
any well typed term of the source language can be encoded into \lpme,
and its type in \lpme is the encoding of its type in the source language.
Finally, we show that a type checker for the \(\lambda\Pi\)-calculus modulo rewriting can be used to type check
terms that have been encoded as described in \cref{sec:encoding}.


%% file: description.tex
\section{\pcert: A Minimal System With Predicate Subtyping}
\label{sec:pcert}

Because of its size, encoding the whole of \pvs cannot be achieved in one step.
Consequently, F.~Gilbert in his PhD~\cite{gilbert2018phd}
extracted, formalised and studied a subsystem of \pvs
which captures the essence of predicate subtyping named \pcert.
Unlike \pvs, \pcert{} contains proof terms,
which has for consequence that type checking is decidable in \pcert{}
while it is not in \pvs.\@
Hence \pcert{} is a good candidate to be a logical system
in which \pvs proofs and specifications can be encoded
to be rechecked by external tools.

In this paper, we use an equational presentation of \pcert{}, that is,
we use equations rather than reduction rules and
slightly change the syntax of terms.
We describe \pcert, as done in~\cite{gilbert2018phd},
namely the addition of predicate subtyping over simple type theory.

\subsection{Type Systems Modulo Theory}\label{sec:pts-modulo}

To describe \pcert and \lpme in a uniform way, we will use the notion of
\emph{Type Systems Modulo} described in \cite{blanqui01phd}. Type Systems
Modulo are an extension of \emph{Pure Type Systems}
~\cite{BarendregtH90} with symbols of fixed arity whose types are
given by a \emph{typing signature} \(\Sigma\), and an arbitrary
conversion relation \(\rlZdeq\) instead of just $\beta$-conversion
\(\equiv_\beta\).

The terms of such a system are characterised by a finite set of
\emph{sorts} \(\mathcal{S}\), a countably infinite set of variables
\(\mathcal{V}\) and a signature \(\Sigma\). The set of terms
\(\mathcal{T}(\Sigma, \mathcal{S}, \mathcal{V})\) is inductively
defined in \cref{fig:terms}.

\begin{figure}[ht]
  \begin{multline*}
    M,N,T,U ::= s \in \mathcal{S} \mid x \in \mathcal{V}
    \mid M~N \mid \teZabst{x}[T]{M} \mid \teZprod{x}{T}{U}
    \mid f(\overrightarrow{M})\\
    \text{with }
    \Sigma(f) = \left(\overrightarrow{x, T}, U, s\right)
  \end{multline*}
  \caption{Terms of the type system characterised by
    \(\mathcal{S}, \mathcal{V}\) and \(\Sigma\).}\label{fig:terms}
\end{figure}

The contexts are noted $\Gamma ::= \varnothing \mid \Gamma, v: T$
and the judgements \(\Gamma \vdash WF \) or \(\Gamma \vdash M : T\).
The typing rules are given in \cref{fig:pts-typing-rules} and depend on
\begin{itemize}
  \item \emph{axioms} \(\mathcal{A}\subseteq\mathcal{S}\times\mathcal{S}\) to type sorts;
  \item \emph{product rules} \(\mathcal{P}\subseteq\mathcal{S}\times\mathcal{S}\times\mathcal{S}\) to type dependent products;
  \item a typing signature \(\Sigma\) which defines the function
    symbols and how to type their applications;
  \item a convertibility relation \(\rlZdeq\).
\end{itemize}

\begin{figure}[ht]
\begin{mathpar}
  \prftree[l]{empty}{}{\varnothing \vdash WF}
  \and
  \prftree[l][r]{$v \not\in \Gamma$}{decl}{\Gamma \vdash T : s}{\Gamma, v : T \vdash WF}
  \and
  \prftree[l][r]{$v: T \in \Gamma$}{var}{\Gamma \vdash WF}{\Gamma \vdash v: T}
  \and
  \prftree[l]{conv}{\Gamma \vdash M: U}{\Gamma \vdash T:
    s}{T \rlZcov U}{\Gamma \vdash M: T}
  \and
  \prftree[l][r]{$(s_{1}, s_{2}) \in \mathcal{A}$}{sort}{\Gamma \vdash
    WF}{\Gamma \vdash s_1: s_2}
  \and
  \prftree[l][r]{$(s_{1}, s_{2}, s_{3}) \in \mathcal{P}$}{prod}{%
    \Gamma \vdash T: s_1
  }{%
    \Gamma, x: T \vdash U: s_2
  }{%
    \Gamma \vdash \teZprod{x}{T}{U}: s_3
  }
  \and
  \prftree[l]{abst}{\Gamma \vdash \teZprod{x}{T}{U}: s}{\Gamma, x: T \vdash M:
    U}{\Gamma \vdash \teZabst{x}[T]{M}: \teZprod{x}{T}{U}}
  \and
  \prftree[l]{app}{\Gamma \vdash M: \teZprod{x}{T}{U}}{\Gamma \vdash N: T}{%
    \Gamma \vdash M~N: \teZsubs{U}{x}{N}
  }
  \and
  \prftree[l][r]{\(\Sigma(f) = \left(\overrightarrow{x, T}, U, s\right)\)}{sig}{%
    \overrightarrow{x: T} \vdash U: s}{%
    \left(\Gamma \vdash t_i: \teZsubsp{T_i}{\left(x_j \mapsto t_j\right)_{j < i}}\right)_i
    }{\Gamma \vdash f(\vec{t}): \teZsubsp{U}{\overrightarrow{x \mapsto t}}}
\end{mathpar}
\caption{Typing rules of a \textsc{Type System Modulo}.}\label{fig:pts-typing-rules}
\end{figure}

\subparagraph*{Notations} Rewriting relations are noted \(\rlZrw_{R}\), where
\(R\) is a set of rewriting rules. \(\rlZrw_{R}\) is the closure of \(R\) by
substitution and context. \(\rlZcov_{R}\) is the symmetric, reflexive and
transitive closure of \(\rlZrw_{R}\). The substitution of \(x\) by \(N\) in
\(M\) is noted \(\teZsubs{M}{x}{N}\). We use a vectorised notation for products
\((\overrightarrow{x: T}) \teZarr U\) to represent the dependent product
\(\teZprodp{x_{1}}{T_{1}}\teZprodp{x_{2}}{T_{2}} \cdots \teZprodp{x_{n}}{T_{n}} U\);
and more generally for any construction that can be extended to a finite
sequence, such as a parallel substitution
\(\teZsubsp{M}{\overrightarrow{x \mapsto N}}\). A mapping
\(\Sigma(f) = (\overrightarrow{x: T}, U, s)\) can also be written
\(\overrightarrow{x: T} \vdash_{\Sigma} f(\vec{x}): U: s\).
For all relations on terms \(R\) and \(S\),
we write \(RS = \{(t, u) \mid \exists v, t R v \land v S u\}\) the composition of \(R\) and \(S\),
and \(R^{*}\) the reflexive and transitive closure of \(R\).

\subsection{Simple Type Theory}

\pvs and \pcert
are both based on simple type theory, which can be represented by the PTS
\lhol{}~\cite{BarendregtH90}:
\begin{itemize}
  \item \(\mathcal{S}^{\text{\lhol}} = \{\teZProp, \teZType, \teZKind\}\),
  \item \(\mathcal{A}^{\text{\lhol}} = \{(\teZProp, \teZType), (\teZType, \teZKind)\}\),
  \item \(\mathcal{P}^{\text{\lhol}} = \{(\teZProp, \teZProp, \teZProp),
    (\teZType, \teZType, \teZType),
    (\teZType, \teZProp, \teZProp)\}\),
  \item \(\Sigma^{\text{\lhol}} = \emptyset\),
  \item \(\rlZdeq^{\text{\lhol}}\) is the reflexive, transitive and symmetric closure of the
    \(\beta\)-equation
    \begin{equation}
      \label{eq:beta}\tag{\(\beta\)}
      ((\teZabst{x}{M})~N) = \teZsubs{M}{x}{N}
    \end{equation}
\end{itemize}

\subsection{Predicate Subtyping}\label{sec:pcert-predicate-subtyping}

Predicate subtyping has two main benefits for a specification language.
The first is to provide a richer type system
thanks to the entanglement of type-checking
and proof-checking.
In consequence, any property can by encoded in the type system,
which allows to easily create ``guards'' such as
\(\mathtt{tail}: \mathtt{nonempty\_stack} \to \mathtt{stack}\)
where \(\mathtt{nonempty\_stack}\) is a predicate subtype
defined from a predicate \texttt{empty?}.
It is also essential in the expression of mathematics:
the judgement \(M: T\) is akin to the statement \(M \in T\)
in the usual language of mathematics when \(T\) is a set defined by
comprehension such as \(E = \{n: \N \mid P(n)\}\).
With predicate subtyping, we can represent the set $E$ by the type
\((\teZpsub{\N}{P})\),
and the judgement
\(\Gamma \vdash M: \teZpsub{\N}{P}\)
is derivable if term \(M\)
contains a proof of \(P(n)\) for some $n$.
The other benefit of predicate subtyping,
which is essential in \pvs developments,
is that it separates the process of writing specifications
from the proving phase.
In \pvs, this separation appears through
\emph{type correctness conditions} (\tcc):
the development of specifications creates proof obligations
that may be solved at any time.
This separation is also visible in usual mathematical developments,
where if we want to prove that \(t \in E\),
we prove once that \(P(t)\) is valid to then forget the proof and simply use
\(t\).

The type system of \pcert{} can be seen as \lhol with a non empty signature
\(\Sigma^{\text{\pvs}}\) defined in \cref{fig:sig-pcert}
and a richer equivalence \(\pcZdeq\) that will be discussed in the next
paragraph.

\begin{figure}[ht]
  \begin{alignat}{3}
    T: \teZType, p: T \teZarr \teZProp
    &\vdash \teZpsub{T}{p}
    &&: \teZType &&: \teZKind\\
    T: \teZType, p: T \teZarr \teZProp, m : T, h: p\, m
    &\vdash \teZpair{T}{p}{m}{h}
    &&: \teZpsub{T}{p} &&: \teZType\\
    T: \teZType, p: T \teZarr \teZProp, m: \teZpsub{T}{p}
    &\vdash \teZprol{T}{p}{m}
    &&: T &&: \teZType\\
    T: \teZType, p: T \teZarr \teZProp, m: \teZpsub{T}{p}
    &\vdash \teZpror{T}{p}{m}
    &&: p\, (\teZprol{T}{p}{m}) &&: \teZType
  \end{alignat}
  \caption{Signature \(\Sigma^{\text{\pvs}}\) of \pcert.}\label{fig:sig-pcert}
\end{figure}

A predicate subtype \((\teZpsub{T}{U})\)
is defined from a \emph{supertype} \(T\) and predicate \(U\)
which binds a variable of type \(T\) to a proposition.
Terms inhabiting a predicate subtype \((\teZpsub{T}{U})\)
are built with the pair construction \((\teZpair{T}{U}{M}{N})\)
where \(M\) is a term of the supertype \(T\) and \(N\) is a proof of
\((U\,M)\).
While the pair construction allows to coerce a term from any type to
a predicate subtype, the converse,
that is the coercion from a type to its supertype is done with
\(\teZprolp\), the left projection of the pair.
The right projection, \(\teZprorp\),
provides a witness that the left projection of the pair validates
the predicate defining the subtype.
Unlike \pcert, PVS does not use coercions \(\teZpairp\), \(\teZprolp\) and
\(\teZprorp\). In PVS, subtyping is implicit: terms do not have a unique
type, and the choice of this type is left to the type checker.

\begin{remark}
  Unlike the original presentation of \pcert in~\cite{gilbert2018phd},
  this one annotates \(\teZprolp\) and \(\teZprorp\), using
  \(\teZprol{T}{p}{m}\) instead of \(\teZprolp\, m\)
  to ease the well-definedness proof of the translation of \pcert terms
  (\cref{prop:termination-translation}).
\end{remark}

\paragraph*{Equations and Proof Irrelevant Pairs}\label{par:pvs-cert-equivalence}

So far, no real difference has been evinced between \pcert and dependent pairs:
predicate subtype \((\teZpsub{T}{p})\) may be encoded as the dependent pair type
\(\Sigma x: T, p\, x\)~\cite[Definition 4.2.3]{gilbert2018phd}.
The difference lies in the equivalence relations and
the fact that \pcert implements \emph{proof irrelevance} in pairs.

The equivalence of \pcert is noted \(\pcZdeq\) and contains
\cref{eq:beta,eq:pcert-eq-pair-pi,eq:pcert-eq-prol}
which provide \emph{proof irrelevance}:
\begin{gather}
  \label{eq:pcert-eq-pair-pi}
  \teZpair{t}{u}{m}{h_{0}} = \teZpair{t}{u}{m}{h_{1}}\\
  \label{eq:pcert-eq-prol}
  \teZprol{t_{0}}{u_{0}}{(\teZpair{t_{1}}{u_{1}}{m}{h})} = m
\end{gather}

We will now motivate the use of these equations in \pcert.
Proofs contained in terms are essential for typing purposes.
On the other hand, these proofs are a burden regarding equivalence of
terms. Were these proofs taken into account
(as \(\rlZcov_{\beta}\) does),
too many terms would be distinguished.
For example, consider two terms \(t = \teZpair{\N}{Even}{2}{h}\)
and \(t' = \teZpair{\N}{Even}{2}{h'}\) typed as even numbers.
Then \(t\) and \(t'\) are not considered equal because they don't have the same
proof (\(h\) and \(h'\)) that 2 is even.
We end up with one even number 2 per proof that 2 is even.

As stated in~\cite{debruijn1994proofirr},
most mathematicians seek convertibility of \(t\) and \(t'\)
and care more about what \(h\) and \(h'\) prove than the proofs themselves.
To this end, \pcert has \emph{proof irrelevant} pairs:
proofs attached to terms are not taken into account when checking the
equivalence of two pairs.
This property is embedded in the equivalence relation \(\pcZdeq\) used in
the conversion rule of \pcert which must verify
\cref{eq:pcert-eq-pair-pi}.

\Cref{eq:pcert-eq-prol} allows the projection to compute,
but because of proof irrelevance, we cannot allow the right projection to
compute, otherwise, all terms of type \(\teZProp\) would be considered
equivalent.

A proof of $T\equiv_\beta U$ or \(T\pcZdeq U\) can use untyped
intermediate terms, which can be problematic when one wants to prove
some property on typed terms only. In the case of $\equiv_\beta$, the
problem is solved by using the fact that $\hookrightarrow_\beta$ is
confluent, that is
${\equiv_\beta}={\hookrightarrow_\beta^*\hookleftarrow_{\beta}^{*}}$. We
now prove a similar property for $\pcZdeq$:

\begin{lemma}[Properties of the \pcert conversion]
  \label{prop:pcert-eq-confluence}
  Let ${\rlZrw_{\beta\teZprolp}} = {\rlZrw_\beta \cup
    \rlZrw_{\teZprolp}}$ where $\rlZrw_{\teZprolp}$ is the closure by
  substitution and context of \cref{eq:pcert-eq-prol} oriented from
  left to right, and let $\leftrightarrow_{pi}$ be the closure by
  substitution and context of \cref{eq:pcert-eq-pair-pi}
  and \(=_{pi} = \leftrightarrow_{pi}^{*}\).

  For all relation on terms $R$, let $R^{ty}$ be the
  restriction of $R$ to typable terms. Then:
  \begin{itemize}
  \item ${\pcZdeq} \subseteq
    {\hookrightarrow_{\beta\teZprolp}^*=_{pi}\hookleftarrow_{\beta\teZprolp}^*}$
  \item $\hookrightarrow_{\beta\teZprolp}$ preserves typing: if $\Gamma\jgZpvs
    M:T$ and $M\hookrightarrow_{\beta\teZprolp}M'$, then $\Gamma\jgZpvs M':T$
  \item ${\pcZdeq^{ty}} \subseteq
    {\left(\rlZrw_{\beta\teZprolp}^{ty}\right)^*
      \left(\leftrightarrow_{pi}^{ty}\right)^*
      \left(\rlZwr_{\beta\teZprolp}^{ty}\right)^*}$,
  \end{itemize}
\end{lemma}

\begin{proof}
  A relation $\hookrightarrow$ is confluent modulo some relation $E$
  if
  ${\hookleftarrow^*\hookrightarrow^*}\subseteq{\hookrightarrow^*E\hookleftarrow^*}$. If
  $E=\emptyset$, we simply say that $\hookrightarrow$ is confluent.

  First note that $\hookrightarrow_{\beta\teZprolp}$ is confluent since
  it can be seen as a Combinatory Reduction System that is orthogonal
  (i.e. whose rules are left-linear and non-overlapping)~\cite{klop93tcs}.

  We now prove that $\leftrightarrow_{pi}$ steps can be postponed:
  ${\leftrightarrow_{pi}\hookrightarrow_{\beta\teZprolp}}\subseteq{\hookrightarrow_{\beta\teZprolp}^==_{pi}}$,
  where $\hookrightarrow_{\beta\teZprolp}^=$ is the reflexive closure of
  $\hookrightarrow_{\beta\teZprolp}$. Assume that the
  $\leftrightarrow_{pi}$ step is at position $p$ and the
  $\hookrightarrow_{\beta\teZprolp}$ step is at position $q$. If $p$ and
  $q$ are disjoint, this is immediate. If $p$ is above $q$, we have
  $\teZpair{T}{U}{M}{N_1}\leftrightarrow_{pi}\teZpair{T}{U}{M}{N_2}$ and
  either $\teZpair{T}{U}{M}{N_2}\hookrightarrow_{\teZprolp} M$ or $\teZpair{T}{U}{M}{N_2}\hookrightarrow_{\beta\teZprolp} \teZpair{T'}{U'}{M'}{N_2'}$.
  In the first case, $\teZpair{T}{U}{M}{N_1}\hookrightarrow_{\teZprolp} M$. In the second case, $\teZpair{T}{U}{M}{N_1}\hookrightarrow_{\beta\teZprolp}^=\teZpair{T'}{U'}{M'}{N_1}$ $\leftrightarrow_{pi}\teZpair{T'}{U'}{M'}{N_2'}$.
  Finally,
  if $q$ is above $p$, we have
  $(\teZabst{x}[T]{M}){N}\leftrightarrow_{pi}(\teZabst{x}[T']{M'}){N'}$
  $\hookrightarrow_{\beta\teZprolp}\teZsubs{M'}{x}{N'}$
  and
  $(\teZabst{x}[T]{M}){N}\hookrightarrow_{\beta\teZprolp}\teZsubs{M}{x}{N}=_{pi}\teZsubs{M'}{x}{N'}$,
  and similarly in the case of a $\teZprolp$ step.
  
  Hence, (1) $\hookrightarrow_{\beta\teZprolp}$ is confluent modulo
  $=_{pi}$, that is,
  ${\pcZdeq} \subseteq {\hookrightarrow_{\beta\teZprolp}^*=_{pi}\hookleftarrow_{\beta\teZprolp}^*}$.

  We now prove that (2) $\hookrightarrow_\beta$ preserves typing. To
  this end, it suffices to prove that, if $\teZprod{x}{T}{U}$ and
  $\teZprod{x}{T'}{U'}$ are typable, and $\teZprod{x}{T}{U}\pcZdeq
  \teZprod{x}{T'}{U'}$, then $T\pcZdeq T'$ and $U\pcZdeq U'$
  (see~\cite{blanqui05mscs} for more details), which follows from (1).

  We now prove that (3) $\hookrightarrow_{\teZprolp}$ preserves
  typing. Assume that $\teZprol{T_{0}}{P_{0}}{(\teZpair{T_{1}}{P_{1}}{M}{N})}$
  is of type
  $C$. By inversion of typing rules,
  $\teZpair{T_{1}}{P_{1}}{M}{N}$ is of type $\teZpsub{T_{0}}{P_{0}}$
  and ${T_{0}}\pcZdeq C$.
  By inversion again, $M$ is of type $T_{1}$ and
  $\teZpsub{T_{0}}{U_{0}} \pcZdeq \teZpsub{T_{1}}{P_{1}}$. By (1),
  $T_{0} \pcZdeq T_{1}$ and $P_{0} \pcZdeq P_{1}$.
  Therefore, $M$ is of type $C$.

  Next, note that (4) ${=_{pi}}={\Leftrightarrow_{pi}}$ where
  $\Leftrightarrow_{pi}$ consists in applying several
  $\leftrightarrow_{pi}$ steps at disjoint
  positions. Indeed, if
  $t=\teZpair{T}{P}{M}{N_1} \leftrightarrow_{pi} u =
  \teZpair{T}{P}{M}{(\ldots \linebreak[1]
    (\teZpair{T'}{P'}{M'}{N_1'})\linebreak[1]\ldots)} \linebreak[1]
  \leftrightarrow_{pi} v =
  \teZpair{T}{P}{M}{(\ldots(\teZpair{T'}{P'}{M'}{N_2'})\ldots)}$,
  then $t\leftrightarrow_{pi}v$ as well.

  Moreover, we have (5)
  ${\Leftrightarrow_{pi}^{ty}}={(\leftrightarrow_{pi}^{ty})^*}$. Indeed,
  $A\Leftrightarrow_{pi}^{ty}B$ means that we can obtain $B$ from $A$
  by replacing some subterms of $A$, that are typable since $A$ is
  typable, by some subterms of $B$, that are typable since $B$ is
  typable.
  
  We can now conclude as follows. Assume that
  $A\pcZdeq^{ty}B$. By (1), there are $A'$ and $B'$ such that
  $A\hookrightarrow_{\beta\teZprolp}^*A'=_{pi}B'\hookleftarrow_{\beta\teZprolp}^*B$. By
  (2), (3), (4) and (5),
  $A{(\hookrightarrow_{\beta\teZprolp}^{ty})^*}A'
  {(\leftrightarrow_{pi}^{ty})^*}B'{(\hookleftarrow_{\beta\teZprolp}^{ty})^*}B$.
\end{proof}


%% file: encoding.tex
\section{Encoding \pcert in \lpme}\label{sec:encoding}

We provide an encoding of \pcert{} into the logical framework \lpme{}.
This encoding allows to express terms of \pcert{} into \lpme.
Because logical frameworks strive to remain minimal,
constructions such as \(\teZpairp\) or \(\teZpsubp\) are not built-in:
they must be expressed into the language of the logical framework
through an encoding.
We hence define the symbols allowing to emulate predicate subtyping
using the terms of \lpme.

\paragraph*{Definition of \lpme}
\lpme is the family of Type Systems Modulo whose sorts, axioms and
product rules are:
\begin{itemize}
  \item sorts
    \(\mathcal{S}^{\mathrm{\lambda\Pi}} = \{\lpZTYPE, \lpZKIND\}\),
  \item axiom \(\mathcal{A}^{\mathrm{\lambda\Pi}} = \{(\lpZTYPE, \lpZKIND)\}\),
  \item product rules \(\mathcal{P}^{\mathrm{\lambda\Pi}} =
    \{(\lpZTYPE, \lpZTYPE, \lpZTYPE),
    (\lpZTYPE, \lpZKIND, \lpZKIND)\}\).
\end{itemize}

\subsection{Encoding Simple Type Theory}\label{sec:encoding-lhol}

The encoding of \lhol given in
\cref{fig:lhol-in-lpme,fig:lhol-in-lpme-eqs}
follows the method settled in~\cite{cousineau07tlca}
for pure type systems.

In the following, we write
the function symbols of a signature in
\textcolor{light-blue-silk}{blue} and the other
constructions of \lpme in black, to better distinguish them.

The general idea is to manipulate types and terms of \lhol
as terms of \lpme.
Sorts are both objectified as \(\ecZtype\) and \(\ecZprop\)
and encoded as types by \(\ecZKind\), \(\ecZType\) and \(\ecZProp\)
in \cref{sym:Kind,sym:Set,sym:Bool,sym:set,sym:bool}.
Sorts as types are used to type sorts as objects to encode
the axioms in \(\mathcal{A}\).
Terms of type \(\teZType\) are encoded
as terms of type \(\ecZType\).
These encoded types can be interpreted as \lpme{} types with
function \(\ecZEl\) \labelcref{sym:El}.
Similarly, propositions are reified as terms of type \(\ecZprop\)
and interpreted by function \(\ecZPrf\).
For instance, given a \lhol type \(T\) and a \lhol proposition \(P\)
both encoded as \lpme terms,
the abstractions \(\teZabst{x}[\ecZEl T]{x}\) and
\(\teZabst{h}[\ecZPrf P]{h}\) are valid \lpme{} terms.
The signature exposed in \cref{fig:lhol-in-lpme} is noted
\(\Sigma^{\text{\lhol}}\).

\cref{eq:eq-lhol-arrd,eq:eq-lhol-impd,eq:eq-lhol-forall} are used to map encoded products
to \lpme products.
\Cref{eq:eq-lhol-bool} makes sure that the objectified
sort \(\ecZprop\) is the same as the sort \(\ecZProp\) when interpreted as a
type.

\begin{figure}[ht]
\begin{alignat}{3}
  &\vdash \ecZKind&&: \lpZTYPE&&: \lpZKIND
  \label{sym:Kind}
  \\
  &\vdash \ecZType&&: \lpZTYPE&&: \lpZKIND
  \label{sym:Set}
  \\
  \label{sym:Bool}
  &\vdash \ecZProp&&: \lpZTYPE&&: \lpZKIND
  \\
  &\vdash \ecZtype &&: \ecZKind&&: \lpZTYPE
  \label{sym:set}
  \\
  &\vdash \ecZprop&&: \ecZType&&: \lpZTYPE
  \label{sym:bool}
  \\
  t: \ecZType &\vdash \ecZEl t&&: \lpZTYPE&&: \lpZKIND
  \label{sym:El}
  \\
  p: \ecZProp &\vdash \ecZPrf p&&: \lpZTYPE&&: \lpZKIND
  \label{sym:Prf}
  \\
  \label{sym:lhol-forall}
  t: \ecZType, p: \ecZEl t \teZarr \ecZProp &\vdash \ecZfa t\, p&&: \ecZProp&&: \lpZKIND
  \\
  \label{sym:lhol-impd}
  p: \ecZProp, q: \ecZPrf p \teZarr \ecZProp &\vdash p \ecZimpd q&&: \ecZProp&&: \lpZKIND
  \\
  \label{sym:lhol-arrd}
  t: \ecZType, u: \ecZEl t \teZarr \ecZType &\vdash t \ecZarrd u&&: \ecZType&&: \lpZKIND
\end{alignat}
\caption{Signature \(\Sigma^{\text{\lhol}}\) of the encoding of \lhol into
  \lpme.}\label{fig:lhol-in-lpme}
\end{figure}

\begin{figure}[ht]
  \begin{align}
    \label{eq:eq-lhol-bool}
    \ecZEl \ecZprop &= \ecZProp
    \\
    \label{eq:eq-lhol-forall}
    \ecZPrf (\ecZfa t\, p) &= \teZprod{x}{\ecZEl t}{\ecZPrf (p\, x)}
    \\
    \label{eq:eq-lhol-impd}
    \ecZPrf (p \ecZimpd q) &= \teZprod{h}{\ecZPrf p}{\ecZPrf (q\, h)}
    \\
    \label{eq:eq-lhol-arrd}
    \ecZEl (t \ecZarrd u) &= \teZprod{x}{\ecZEl t}{\ecZEl (u\, x)}
  \end{align}
  \caption{Equations of the encoding of \lhol into \lpme.}\label{fig:lhol-in-lpme-eqs}
\end{figure}

\subsection{Encoding Predicate Subtyping}

Predicate subtypes are defined in \cref{eq:encoding-psub} as encoded types
(i.e.\ terms of type \(\ecZType\)) built
from encoded type \(t\) and predicate defined on \(t\).
Pairs are encoded in \cref{eq:encoding-pair},
where the second argument is the predicate that defines the type of the pair.
The two projections are encoded in
\cref{eq:encoding-proj-l,eq:encoding-proj-r},
and we note the signature of \cref{fig:pcert-encoding}
\(\Sigma^{\ecZpsub}\).

\begin{figure}[ht]
\begin{alignat}{3}
  \label{eq:encoding-psub}
  t: \ecZType, p: \ecZEl t \teZarr \ecZProp
  &\vdash
  \ecZpsub t\, p &&: \ecZType &&: \lpZTYPE
  \\
  \label{eq:encoding-pair}
  t: \ecZType, p: \ecZEl t \teZarr \ecZProp, m: \ecZEl t, h: \ecZPrf (p\, m)
  &\vdash
  \ecZpair t\, p\, m\, h
  &&: \ecZEl (\ecZpsub t\, p)
  &&: \lpZTYPE
  \\
  \label{eq:encoding-proj-l}
  t: \ecZType, p: \ecZEl t \teZarr \ecZProp, m: \ecZEl (\ecZpsub t\,p)
  &\vdash
  \ecZfst t\,p\,m &&: \ecZEl t &&: \lpZTYPE
  \\
  \label{eq:encoding-proj-r}
  t: \ecZType, p: \ecZEl t \teZarr \ecZProp, m: \ecZEl (\ecZpsub t\,p)
  &\vdash \ecZsnd t\,p\,m &&: \ecZPrf (p\,(\ecZfst t\,p\,m)) &&: \lpZTYPE
\end{alignat}
\caption{Signature \(\Sigma^{\ecZpsub}\) of the encoding of predicate
  subtyping into \lpme.}\label{fig:pcert-encoding}
\end{figure}

The signature used to encode \pcert into \lpme
is \(\sigPcertEnc =
\Sigma^{\text{\lhol}} \cup \Sigma^{\ecZpsub}\).
The terms of the encoding are thus the terms of
\(\mathcal{T}(\sigPcertEnc,
\mathcal{S}^{\mathrm{\lambda\Pi}}, \mathcal{V})\).
The typing rules are those of \lpme with the signature
\(\sigPcertEnc\) and
the congruence \(\lpZdeq\) generated by
\cref{eq:pcert-eq-pair-pi,eq:pcert-eq-prol,eq:beta,eq:eq-lhol-bool,%
  eq:eq-lhol-forall,eq:eq-lhol-impd,eq:eq-lhol-arrd}
where, in \cref{eq:pcert-eq-pair-pi,eq:pcert-eq-prol},
\(\teZpsubp\), \(\teZpairp\) and \(\teZprolp\) (\pcert symbols in black)
are replaced by \(\ecZpsub, \ecZpair\) and \(\ecZfst\) (\lpme symbols in blue).

\subsection{Translation of \pcert Terms Into \lpme Terms}

\begin{definition}[Translation]
  Let \(\Gamma\) be a well formed context.
  \begin{itemize}
    \item The \emph{term translation} of the terms \(M\)
      typable in $\Gamma$, noted \(\trZobj{M}_{\Gamma}\),
      is defined in \cref{fig:lhol-to-lpme,fig:pcert-to-lpme}.

    \item The \emph{type translation} of \(\teZKind\) and the terms
      \(M\) typable by a sort in $\Gamma$,
      noted \(\trZtyp{M}_{\Gamma}\),
      is defined in \cref{fig:type-translation}.

    \item The \emph{context translation}
      \(\trZctx{\Gamma}\) is defined by induction on \(\Gamma\) as
      \begin{equation*}
        \trZctx{\emptyset} = \emptyset; \quad
        \trZctx{\Gamma, x: T} = \trZctx{\Gamma}, x: \trZtyp{T}_{\Gamma}
      \end{equation*}
  \end{itemize}
  \end{definition}

\begin{proposition}\label{prop:termination-translation}
  The translation function \(\trZobj{\cdot}_{\cdot}\) that maps a
  context and a \pcert{} term typable in this context to a \lpme{}
  term is well-defined.
\end{proposition}

\begin{proof}

  After \cref{prop:pcert-eq-confluence} and~\cite[Lemma
    41]{blanqui01phd}, the types of a term are unique up to
  equivalence. Moreover, the arguments of the translation function are
  decreasing with respect to the (strict) subterm relation.
\end{proof}

\begin{figure}[ht]
  \begin{align*}
    \trZobj{x}_{\Gamma} &= x \\
    \trZobj{\teZProp}_{\Gamma} &= \ecZprop \\
    \trZobj{\teZType}_{\Gamma} &= \ecZtype \\
    \trZobj{M\, N}_{\Gamma} &=
    \trZobj{M}_{\Gamma}\ \trZobj{N}_{\Gamma} \\
    \trZobj{\teZabst{x}[T]{M}}_{\Gamma} &=
    \teZabst{x}[\ecZEl~\trZobj{T}_{\Gamma}]{\trZobj{M}_{\Gamma, x: T}}\\
    \trZobj{\teZprod{x}{T}{U}}_{\Gamma} &=
    \trZobj{T}_{\Gamma} \ecZarrd
    \left(\teZabst{x}[\trZtyp{T}_{\Gamma}]{\trZobj{U}_{\Gamma,x: T}}\right)\\
    &\phantom{=} \text{when } \Gamma \jgZpvs T: \teZType \text{ and }
    \Gamma,x: T \jgZpvs U: \teZType \\
    \trZobj{\teZprod{x}{T}{P}}_{\Gamma} &=
    \ecZfa \trZobj{T}_{\Gamma}\,
    \left(\teZabst{x}[\trZtyp{T}_{\Gamma}]{\trZobj{P}_{\Gamma, x: T}}\right)\\
    &\phantom{=} \text{when } \Gamma \jgZpvs T: \teZType
    \text{ and } \Gamma, x: T \jgZpvs P: \teZProp\\
    \trZobj{\teZprod{h}{P}{Q}}_{\Gamma} &=
    \trZobj{P}_{\Gamma} \ecZimpd
    \left(\teZabst{h}[\trZtyp{P}_{\Gamma}]{\trZobj{Q}_{\Gamma, h: P}}\right)\\
    &\phantom{=} \text{when } \Gamma \jgZpvs P: Prop
    \text{ and } \Gamma, h: P \jgZpvs Q: \teZProp \\
  \end{align*}
  \caption{Translation from \lhol to \lpme.}\label{fig:lhol-to-lpme}
\end{figure}

\begin{figure}
  \begin{minipage}{0.5\textwidth}
  \begin{align*}
    \trZobj{\teZpsub{T}{P}}_{\Gamma} &=
    \ecZpsub~\trZobj{T}_{\Gamma}\, \trZobj{P}_{\Gamma}
    \\
    \trZobj{\teZpair{T}{P}{M}{N}}_{\Gamma} &=
    \ecZpair \trZobj{T}_{\Gamma}\, \trZobj{P}_{\Gamma}\,
    \trZobj{M}_{\Gamma}~\trZobj{N}_{\Gamma}
  \end{align*}
  \end{minipage}
  \begin{minipage}{0.5\textwidth}
    \begin{align*}
    \trZobj{\teZprol{T}{P}{M}}_{\Gamma}
      &=
        \ecZfst \trZobj{T}_{\Gamma}\, \trZobj{P}_{\Gamma}\,
        \trZobj{M}_{\Gamma}
    \\
    \trZobj{\teZpror{T}{P}{M}}_{\Gamma}
      &= \ecZsnd \trZobj{T}_{\Gamma}\, \trZobj{P}_{\Gamma}\, \trZobj{M}_{\Gamma}
    \end{align*}
  \end{minipage}
  \caption{Translation from \pcert to \lpme.}\label{fig:pcert-to-lpme}
\end{figure}

\begin{figure}
  \begin{minipage}{0.5\textwidth}
    \begin{align*}
      \trZtyp{T}_{\Gamma} &= \ecZEl\ \trZobj{T}_{\Gamma} & \text{when } \Gamma \jgZpvs T:\teZType;\\
      \trZtyp{T}_{\Gamma} &= \ecZPrf\ \trZobj{T}_{\Gamma} & \text{when } \Gamma \jgZpvs T:\teZProp;
    \end{align*}
  \end{minipage}
  \begin{minipage}{0.5\textwidth}
    \begin{align*}
      \trZtyp{\teZKind} &= \ecZKind\\
      \trZtyp{\teZType} &= \ecZType
    \end{align*}
  \end{minipage}
  \caption{Translation of types from \pcert to
    \lpme.}\label{fig:type-translation}
\end{figure}

\subsection{Examples of Encoded Theories}

We provide here some examples that take advantage of proof irrelevance
or predicate subtyping.
While these examples could have been presented in \pcert,
we unfold them into the encoding of \pcert into \lpme to show how it can
be used in practice.
All examples are available as \Dedukti
files\footnote{directory \texttt{paper} of
\url{https://github.com/Deducteam/personoj} rev.\ 9807710 (published on Feb. 7, 2021)}
and can be type-checked with
\toolstyle{Lambdapi}\footnote{\url{https://github.com/Deducteam/lambdapi},
  rev.\ 9a90b1be (published on Feb. 18, 2021)}.
In the examples, the first two arguments of
\(\ecZfst\), \(\ecZpair\) and \(\ecZsnd\) are implicit.

\begin{example}[Stacks with predicate subtypes]
This example comes from the language reference manual of
\pvs~\cite{PVS:language} and illustrates the use of predicate
subtyping and the generation of \tcc
through a specification of stacks in \cref{fig:stacks}.
{
  \newcommand\stack{\slcol{\mathrm{stack}}}
  \renewcommand\empty{\slcol{\mathrm{empty}}}
  \newcommand\nonemptyStackP{\slcol{\operatorname{nonempty\_stack?}}}
  \newcommand\nonemptyStack{\slcol{\mathrm{nonempty\_stack}}}
  \newcommand\stackPush{\operatorname{\slcol{push}}}
  \newcommand\stackPop{\operatorname{\slcol{pop}}}
  \newcommand\pushTopPop{\operatorname{\slcol{push\_top\_pop}}}
  \newcommand\popPush{\operatorname{\slcol{pop\_push}}}
  \newcommand\poptPusht{\operatorname{\slcol{pop2push2}}}
  \newcommand\semi{\exZsemi}

  Predicate subtyping is used to define the type of nonempty stacks,
  which allows the function \(\stackPop\) to be total.
  Symbol \(\popPush\) is an axiom that uses Leibniz equality \(=\) on stacks.
  In the definition of the theorem \(\poptPusht\),
  term \(?_{0}\) is a meta-variable that must be instantiated with a proof
  that the first argument of the pair is not empty, and represents,
  in the encoding, the \tcc generated by \pvs.
  We can thus see that the concept of \tcc of \pvs has a clear and explicit
  representation in the encoding, allowing its benefits to be transported to
  \lpme.
  \begin{figure}
    \begin{gather*}
      \exZsym \stack: \ecZType \semi \hspace{6em}
      \exZsym \empty: \ecZEl \stack \semi \hspace{6em}
      \exZsym t: \ecZType \semi\\
      \exZsym \nonemptyStackP (s: \ecZEl \stack) \coloneqq s \ne \empty \semi \\
      \exZsym \nonemptyStack \coloneqq \ecZpsub \nonemptyStackP \semi \\
      \exZsym \stackPush: \ecZEl \stack \teZarr \ecZEl t \teZarr \ecZEl
      \nonemptyStack \semi\\
      \exZsym \stackPop: \ecZEl \nonemptyStack \teZarr \ecZEl \stack \semi \\
      \exZsym \popPush (x: \ecZEl t) (s: \ecZEl \stack) : \ecZPrf (\stackPop (\stackPush x\, s) = s)
      \semi
      \\
      \begin{aligned}
        \exZsym& \poptPusht (x\, y: \ecZEl t) (s: \ecZEl \stack)\\
              :& \ecZPrf (\stackPop (\ecZpair\, (\stackPop (\stackPush x\,
              (\ecZfst (\stackPush y\, s))))\, ?_{0}) = s) \coloneqq \dots \semi
      \end{aligned}
    \end{gather*}

    \caption{Specification for stacks.}\label{fig:stacks}
  \end{figure}
}
\end{example}

\begin{example}[Bounded lists and proof irrelevance]
{
  \newcommand\zero{\exZzero}
  \newcommand\suc{\exZsuc}
  \renewcommand\symbol{\exZsym}
  \newcommand\semi{\exZsemi}
  \newcommand\bound{\mathrm{\slcol{bound}}}
  \newcommand\bounded{\operatorname{\slcol{bounded}}}
  \newcommand\blist{\operatorname{\slcol{blist}}}
  \newcommand\bnil{\operatorname{\slcol{bnil}}}
  \newcommand\bcons{\operatorname{\slcol{bcons}}}

This example is inspired by sorted lists in the \toolstyle{Agda} manual
\cite{agdaManual}\footnote{\url{https://agda.readthedocs.io/en/v2.5.4/language/irrelevance.html}}.
Because we have not encoded dependent types,
we cannot encode the type of lists bounded by a variable.
We thus declare the bound in the signature.
The specification is given in \cref{fig:sorted-lists-spec}.

We first notice that the predicate subtype allows to encode
the proof \(head \leq \bound\) passed as a standalone argument in
\toolstyle{Agda} in the type of an argument in our encoding,
providing a shorter type for \(\bcons\).
  \begin{figure}
    \begin{minipage}{0.4\textwidth}
      \begin{gather*}
        \symbol \zero: \ecZEl \N \semi\\
        \symbol \suc (n: \ecZEl \N): \ecZEl \N \semi\\
        \symbol \leq (n\, m: \ecZEl \N): \ecZProp \semi
      \end{gather*}
    \end{minipage}
    \begin{minipage}{0.6\textwidth}
      \begin{gather*}
        \symbol \bound \coloneqq \ldots\semi\\
        \symbol \blist: \ecZType\semi\\
        \symbol \bnil: \ecZEl \blist\semi
      \end{gather*}
    \end{minipage}
    \begin{gather*}
      \symbol \bounded \coloneqq \ecZpsub(\teZabst{n}{n \leq \bound})\semi\\
      \symbol \bcons (head: \ecZEl \bounded)
      (tail: \ecZEl \blist): \ecZEl \blist\semi
    \end{gather*}
    \caption{Specification of sorted lists.}\label{fig:sorted-lists-spec}
  \end{figure}
In \cref{fig:sorted-lists-instances},
we define two (non-convertible) axioms \(p_{1}\) and \(p_{2}\) as proofs of
\(\zero \leq \suc \bound\),
and two lists containing \(\zero\) but proved to be bounded by \(\suc \bound\)
using \(p_{1}\) for \(\ell_{1}\) and \(p_{2}\) for \(\ell_{2}\).
Type checking \(\ell_{i}\) requires axioms \(p_{i}\).
These axioms are like \tcc's in \pvs.
  \begin{figure}
    \begin{minipage}{0.4\textwidth}
      \begin{gather*}
        \symbol p_{1}: \ecZPrf (\zero \leq \suc \bound)\semi\\
        \symbol p_{2}: \ecZPrf (\zero \leq \suc \bound)\semi
      \end{gather*}
    \end{minipage}
    \begin{minipage}{0.6\textwidth}
      \begin{gather*}
        \symbol \ell_{1} \coloneqq \bcons (\ecZpair \zero~p_{1})
        \bnil \semi\\
        \symbol \ell_{2} \coloneqq \bcons (\ecZpair \zero~p_{2}) \bnil \semi
      \end{gather*}
    \end{minipage}
    \caption{Definition of two sorted lists with different
      proofs.}\label{fig:sorted-lists-instances}
  \end{figure}
Assuming that one wants to prove \(\ell_{1} = \ell_{2}\),
had we lacked proof irrelevance, we would have had to prove that
\(p_{1} \rlZcov p_{2}\), which is not possible.
In our case, the equality is simply the result of
\(\operatorname{refl} \ell_{1}\).
}
\end{example}

%% file: encoding_properties.tex
\section{Correctness of the Encoding}\label{sec:translation}

In this section, we prove that the encoding is correct: if a \pcert
type is inhabited then its translation is inhabited too. Any
type-checker for \lpme could thus be used to recheck \pcert
typings. However, to make sure that our encoding is faithful
(the encoding that maps any \pcert term to the same well-typed ground term
is correct, but useless),
completeness (also called conservativity)
ought to be proved too: a \pcert type is inhabited
whenever its encoding is inhabited. However, as completeness is often
difficult to establish (see \cite{assaf15phd,thire20phd}), we leave it
for future work.

In the following,
\begin{itemize}
  \item \(s\) stands for \(\teZType\), \(\teZProp\) or \(\teZKind\);
  \item \(T, U\) designate terms of type \(Type\);
  \item \(M, N, t, u\) designate expressions that have a type \(T: Type\);
  \item \(P, Q\) are propositions of type \(\teZProp\), or predicates of type
    \(T \teZarr \teZProp\);
  \item \(h\) stands for a proof typed by a proposition.
\end{itemize}
Typing judgements in \pcert are noted with \(\jgZpvs\), and typing judgements
in \lpme are noted with \(\jgZlpme\).

\begin{lemma}[Preservation of substitution]%
  \label{prop:preservation-substitution}
  If \(\Gamma, x: U, \Delta \jgZpvs M: T\) and \(\Gamma \jgZpvs N: T\),
  then
  \(\trZobj{
    \teZsubs{M}{x}{N}
  }_{\Gamma, \{x \mapsto N\} \Delta} =
  \teZsubs{\trZobj{M}_{\Gamma, x:U, \Delta}}{x}{\trZobj{N}_{\Gamma}}\).
\end{lemma}
\begin{proof}
  By structural induction on \(M\).
\end{proof}

\begin{lemma}[Preservation of equivalence]%
  \label{prop:preservation-equivalence}
  Let \(M\) and \(N\) be two well typed terms in \(\Gamma\).
  \begin{enumerate}
    \item\label{it:preservation-single-step}
    If \(M \pcZdeqs N\), then
    \(\trZobj{M}_{\Gamma} \lpZdeq \trZobj{N}_{\Gamma}\).
    \item If \(M \pcZdeq N\), then
      \(\trZobj{M}_{\Gamma} \lpZdeq \trZobj{N}_{\Gamma}\).
  \end{enumerate}
\end{lemma}
\begin{proof}
  Each item is proved separately.
  \begin{enumerate}
    \item Taking back the notations of the proof of
    \cref{prop:pcert-eq-confluence}, we show that
      \begin{enumerate}
        \item computational steps of \(\rlZrw_{\beta \teZprolp}^{ty}\)
          are preserved,
        \item equational steps of \(\leftrightarrow_{pi}^{ty}\) are preserved.
      \end{enumerate}
      These two properties are shown by induction on a context \(C\)
      such that \(M = C[\hat{M}] ~ R ~ C[\hat{N}] = N\)
      where \(R\) is any of the two relations applied
      at the head of \(\hat{M}\) and \(\hat{N}\). We will only detail the
      base cases of inductions, the other cases being straightforward.

      \begin{description}
        \item[Preservation of Computation]
          There are two possible cases,
      \begin{description}
        \item[Case \(M = ((\teZabst{x}{t})~u) \rlZrw_{\beta}
          \teZsubs{t}{x}{u}\),]
          we have,
          \[
          \trZobj{(\teZabst{x}[U]{t})~u}_{\Gamma} =
          ((\teZabst{x}[\trZtyp{U}_{\Gamma}]{\trZobj{t}_{\Gamma, x:U}})\,
          \trZobj{u}_{\Gamma}) =
          \teZsubs{\trZobj{t}_{\Gamma}}{x}{\trZobj{u}_{\Gamma}} \lpZdeq
          \trZobj{\teZsubs{t}{x}{u}}_{\Gamma}
          \]
          where the equivalence is given by
          \cref{prop:preservation-substitution}.

        \item[Case \(M = \teZprol{T_{1}}{P_{1}}{(\teZpair{T_{0}}{P_{0}}{t}{h})}
          \rlZrw_{\teZprolp} t\),]
          we have the following equalities
          \begin{equation*}
            \begin{aligned}
              \trZobj{\teZprol{T_{1}}{P_{1}}{(\teZpair{T_{0}}{P_{0}}{t}{h})}}_{\Gamma}
              &= \ecZfst \trZobj{T_{1}}_{\Gamma}\, \trZobj{P_{1}}_{\Gamma}\,
              \trZobj{\teZpair{T_{0}}{P_{0}}{t}{h}}_{\Gamma}\\
              &
              = \ecZfst \trZobj{T_{1}}_{\Gamma}\, \trZobj{P_{1}}_{\Gamma}
              \left(\ecZpair \trZobj{T_{0}}_{\Gamma}\, \trZobj{P_{0}}_{\Gamma}\,
                \trZobj{t}_{\Gamma}\, \trZobj{h}_{\Gamma}
              \right)
              \\
              &\lpZdeq \trZobj{t}_{\Gamma}
            \end{aligned}
          \end{equation*}
          with the last equivalence provided by \cref{eq:pcert-eq-prol}.
      \end{description}

      \item[Preservation of Proof Irrelevance]
        Assume that
        \(M = \teZpair{T}{P}{t}{h} \leftrightarrow_{pi} \teZpair{T}{P}{t}{h'}\)
        \[
            \trZobj{\teZpair{T}{P}{t}{h}}_{\Gamma} =
            \ecZpair \trZobj{T}_{\Gamma}\, \trZobj{P}_{\Gamma}\,
            \trZobj{t}_{\Gamma}\, \trZobj{h}_{\Gamma}
            \lpZdeq
            \ecZpair \trZobj{T}_{\Gamma}\, \trZobj{P}_{\Gamma}\,
            \trZobj{t}_{\Gamma}\, \trZobj{h'}_{\Gamma} =
            \trZobj{\teZpair{T}{P}{t}{h'}}_{\Gamma}
        \]
        where the equivalence is given by \cref{eq:pcert-eq-pair-pi}.
    \end{description}

    \item
    By \cref{prop:pcert-eq-confluence}, we know that there are \(H_{0}\)
    and \(H_{1}\) such that
    \(M (\rlZrw_{\beta\teZprolp}^{ty})^{*}
    H_{0} (\leftrightarrow_{pi}^{ty})^{*} H_{1}\)
    \({(\rlZwr_{\beta\teZprolp}^{ty})^{*}} N\).
    For \(R \in \{\leftrightarrow_{pi}, \rlZrw_{\beta\teZprolp}\}\),
    we have \(t (R^{ty})^{*} u \Rightarrow \trZobj{t} \lpZdeq \trZobj{u}\)
    by induction on the number of \(R^{ty}\) steps,
    using \cref{it:preservation-single-step} for the base case.
    Therefore,
    \(\trZobj{M}_{\Gamma} \lpZdeq \trZobj{H_{0}}_{\Gamma} \lpZdeq
    \trZobj{H_{1}}_{\Gamma} \lpZdeq \trZobj{N}_{\Gamma}\),
    which gives, by transitivity of \(\lpZdeq\), \(\trZobj{M}_{\Gamma} \lpZdeq \trZobj{N}_{\Gamma}\).
  \end{enumerate}
\end{proof}

\begin{theorem}[Correctness]
  If \(\Gamma \jgZpvs M: T\), then
  \(\trZctx{\Gamma} \jgZlpme \trZobj{M}_{\Gamma}: \trZtyp{T}_{\Gamma}\).
  For all \(\Gamma\), if \(\Gamma \jgZpvs WF\),
  then \(\trZctx{\Gamma} \jgZlpme WF\).
\end{theorem}

\begin{proof}
  By induction on the typing derivation of \(\Gamma \jgZpvs M: T\)
  and case distinction on the last inference rule.
  \begin{description}
    \item[empty]
      \(\prftree{\emptyset \jgZpvs WF}\)

      We have \(\trZctx{\emptyset} = \emptyset\) and
      \(\emptyset \jgZlpme WF\).

    \item[decl]
      \(\prftree[r]{$v \not\in \Gamma$}{\Gamma \jgZpvs T: s}{\Gamma, v: T \jgZpvs WF}\)

      We have \(\trZctx{\Gamma, v: T} = \trZctx{\Gamma}, v: \trZtyp{T}_{\Gamma}\).
      By induction hypothesis,
      we have \(\trZctx{\Gamma} \jgZlpme
      \trZobj{T}_{\Gamma}: \trZtyp{s}_{\Gamma}\),
      for \(s \in \mathcal{S}\) and hence \(\trZtyp{s}_{\Gamma}\) is either
      \(\ecZProp\) by conversion (because \(\ecZEl \ecZprop \lpZdeq \ecZProp\)),
      \(\ecZType\) or \(\ecZKind\).
      If \(s\) is \(\teZKind\), then \(T\) is \(\teZType\).
      Since \(\trZctx{\Gamma} \jgZlpme \ecZType: \lpZTYPE\) because
      \(\sigPcertEnc(\ecZType) = (\vec{0}, (\lpZTYPE, \lpZKIND))\),
      we can derive with the declaration rule
      \(\trZctx{\Gamma, v: T} \jgZlpme WF\) because
      \(\trZtyp{\teZType} = \ecZType\).
      Otherwise, \(s\) is \(\teZType\) or \(\teZProp\) and
      \(\trZtyp{T} = \xi\, \trZobj{T}_{\Gamma}\)
      where \(\xi\) is \(\ecZEl\) or \(\ecZPrf\).
      By typing of \(\ecZEl\) or \(\ecZPrf\) (with the signature),
      \(\trZctx{\Gamma} \jgZlpme \trZtyp{T}_{\Gamma}: \lpZTYPE\)
      and finally,
      \(\trZctx{\Gamma, v: T} \jgZlpme WF\)
      by application of the declaration rule.

    \item[var]
      \(\prftree[r]{$v: T \in \Gamma$}{\Gamma \jgZpvs WF}{\Gamma \jgZpvs v: T}\)

      By definition, \(\trZobj{v} = v\) and by induction hypothesis,
      \(\trZctx{\Gamma} \jgZlpme WF\). Since \(v: T \in \Gamma\),
      by definition, there is \(\Delta \subsetneq \Gamma\),
      \(\Delta \jgZpvs WF\) such that,
      \(v: \trZtyp{T}_{\Delta} \in \trZctx{\Gamma}\).
      Hence \(\trZctx{\Gamma} \jgZlpme v: \trZtyp{T}_{\Delta}\)
      and finally \(\trZctx{\Gamma} \jgZlpme v: \trZtyp{T}_{\Gamma}\)
      because contexts are well formed.

    \item[sort]
      \(\prftree[r]{$(s_{1}, s_{2}) \in \mathcal{A}$}{\Gamma \jgZpvs WF}{\Gamma \jgZpvs s_{1}: s_{2}}\)

      First, \(\trZobj{s_{1}}\) is either \(\ecZprop\) or \(\ecZtype\).
      In the former case, \(\trZtyp{s_{2}} = \ecZType\)
      and because \(\trZctx{\Gamma} \jgZlpme WF\)
      (by induction hypothesis) and
      \(\sigPcertEnc(\ecZprop) = (\vec{0}, (\ecZType, \lpZTYPE))\),
      we have \(\trZctx{\Gamma} \jgZlpme \ecZprop: \ecZType\).
      The same procedure holds for
      \(s_{1} = \teZType\) and \(s_{2} = \teZKind\).

    \item[prod]
      \(\prftree[r]{$(s_{1}, s_{2}, s_{3}) \in \mathcal{P}$}{%
      \Gamma \jgZpvs T: s_{1}}{\Gamma, x: T \jgZpvs U: s_{2}}{%
      \Gamma \jgZpvs \teZprod{x}{T}{U}: s_{3}}\)

      We only detail for the product \((\teZType, \teZProp, \teZProp)\),
      others being processed similarly.
      We have \(\trZobj{\teZprod{x}{T}{U}}_{\Gamma} =
      \ecZfa \trZobj{T}_{\Gamma}\,\left(\teZabst{x}[\trZtyp{T}_{\Gamma}]{%
      \trZobj{U}_{\Gamma, x: T}}\right)\).
      By induction hypothesis,
      \(\trZctx{\Gamma} \jgZlpme \trZobj{T}: \trZtyp{\teZType}\),
      and thus
      \(\trZctx{\Gamma} \jgZlpme \trZobj{T}: \ecZType\) by definition.
      By induction hypothesis,
      \(\trZctx{\Gamma, x: T} \jgZlpme \trZobj{U}: \trZtyp{\teZProp}\),
      and thus \(\trZctx{\Gamma}, x: \trZtyp{T}_{\Gamma}
      \jgZlpme \trZobj{U}: \ecZProp\)
      by definition of \(\trZctx{\cdot}\) and conversion which yields
      \(\trZctx{\Gamma} \jgZlpme
      \teZabst{x}[\trZtyp{T}_{\Gamma}]{\trZobj{U}_{\Gamma, x: T}}:
      \trZtyp{T}_{\Gamma} \teZarr \ecZProp\).

      To finish, we obtain
      \(\trZctx{\Gamma} \jgZlpme \teZabst{x}[\trZtyp{T}_{\Gamma}]{\trZobj{U}_{\Gamma, x: T}}:
      (\ecZEl \trZobj{T}_{\Gamma}) \teZarr \ecZProp\) by conversion.
      Using the typing signature \(\sigPcertEnc\),
      \(\trZctx{\Gamma} \jgZlpme \ecZfa \trZobj{T}_{\Gamma}\,
      \left(\teZabst{x}{\trZtyp{T}_{\Gamma}}{\trZobj{U}_{\Gamma, x: T}}\right): \ecZProp\)
      which becomes, by conversion  \(\ecZProp \lpZdeq \ecZEl \ecZprop\)
      and definition of \(\trZtyp{\cdot}_{\Gamma}\): \(\ecZEl \ecZprop =
      \trZtyp{\teZProp}\), hence,
      \(\trZctx{\Gamma} \jgZlpme \ecZfa \trZobj{T}_{\Gamma}\,
      \left(\teZabst{x}{\trZtyp{T}_{\Gamma}}{\trZobj{U}_{\Gamma, x: T}}\right): \trZtyp{\teZProp}\)

    \item[abst]
      \(\prftree{\Gamma, v: T \jgZpvs M: U}{\Gamma \jgZpvs \teZprod{v}{T}{U}: s}{
      \Gamma \jgZpvs \teZabst{v}[T]{M}: \teZprod{v}{T}{U}}\)

      We have \(\trZobj{\teZabst{v}[T]{M}}_{\Gamma} =
      \teZabst{v}[\trZtyp{T}_{\Gamma}]{\trZobj{M}_{\Gamma}}\).
      By induction hypothesis,
      \(\trZctx{\Gamma, v: T} \jgZlpme \trZobj{M}_{\Gamma, v: T}:
      \trZtyp{U}_{\Gamma, v: T}\) and by definition of \(\trZctx{\cdot}\),
      \(\trZctx{\Gamma}, v: \trZtyp{T}_{\Gamma} \jgZlpme
      \trZobj{M}_{\Gamma, v: T}: \trZtyp{U}_{\Gamma, v: T}\).
      Applying the abstraction rule in \lpme, we obtain
      \(\trZctx{\Gamma} \jgZlpme
      \teZabst{v}[\trZtyp{T}_{\Gamma}]{\trZobj{M}_{\Gamma, v: T}}:
      \teZprod{v}{\trZtyp{T}_{\Gamma}}{\trZtyp{U}_{\Gamma, v: T}}\)
      (with the product well typed in \lpme since
      \(\trZtyp{U}\) and \(\trZtyp{T}\) are both of type \(\lpZTYPE\)
      and thus the product is of type \(\lpZTYPE\) as well).

      Finally, we proceed by case distinction
      on sorts \(s_{T}\) and \(s_{U}\) such that
      \(\Gamma \jgZpvs T: s_{T}\) and \(\Gamma \jgZpvs U: s_{U}\).
      We will detail the case \((s_{T}, s_{U}) = (\teZType, \teZProp)\).
      We have \(\teZprod{v}{\trZtyp{T}_{\Gamma}}{\trZtyp{U}_{\Gamma, v: T}}
      \lpZdeq \ecZPrf (\ecZfa \trZobj{T}_{\Gamma}\,
      (\teZabst{x}[\trZtyp{T}_{\Gamma}]{\trZobj{U}_{\Gamma, v: T}}))
      = \trZtyp{\teZprod{v}{T}{U}}_{\Gamma} \) which allows to conclude.
    \item[app]
      \(\prftree{\Gamma \jgZpvs M: \teZprod{v}{T}{U}}{\Gamma \jgZpvs N: T}{
        \Gamma \jgZpvs M\, N: \teZsubs{U}{v}{N}
      }\)

      By induction hypothesis and conversion, we have
      \(\trZctx{\Gamma} \jgZlpme \trZobj{M}_{\Gamma}:
      \teZprod{v}{\trZtyp{T}_{\Gamma}}{\trZtyp{U}_{\Gamma, v: T}}\)
      (shown by case distinction on the sorts of \(T\) and \(U\))
      and
      \(\trZctx{\Gamma} \jgZlpme \trZobj{N}_{\Gamma}: \trZtyp{T}_{\Gamma}\).
      Since \(\trZobj{M\, N}_{\Gamma} = \trZobj{M}\, \trZobj{N}\),
      we obtain using the application rule
      \(\trZctx{\Gamma} \jgZlpme \trZobj{M\, N}:
      \teZsubs{\trZtyp{U}_{\Gamma, v: T}}{v}{\trZobj{N}_{\Gamma}}\)
      and by \cref{prop:preservation-substitution}, we obtain
      \(\trZctx{\Gamma} \jgZlpme \trZobj{M\, N}:
      \trZtyp{\teZsubs{U}{v}{N}}_{\Gamma}\).

    \item[conv]
      \(\prftree{\Gamma \jgZpvs M: U}{\Gamma \jgZpvs T: s}{T \pcZdeq U}{%
      \Gamma \jgZpvs M: T}\)

      By hypothesis, there is a type \(U\) such that \(\Gamma \jgZpvs M: U\),
      and \(T \pcZdeq U\), and there is a sort \(s\) such that
      \(\Gamma \jgZpvs T: s\).
      By induction hypothesis,
      \(\trZctx{\Gamma} \jgZlpme \trZobj{M}_{\Gamma}:
      \trZtyp{U}_{\Gamma}\).

      We now prove that if \(T \pcZdeq U\), then \(\trZtyp{T}_{\Gamma} \lpZdeq
      \trZtyp{U}_{\Gamma}\) and \(\Gamma \jgZlpme \trZtyp{T}: \lpZTYPE\):
      it will allow us to conclude using the conversion
      rule in \lpme.

      By \cref{prop:pcert-eq-confluence}, we have
      \(T \rlZrw_{\beta\teZprolp}^{*} T' =_{pi} U' \rlZwr_{\beta\teZprolp}^{*}
      U\) and
      \(T (\rlZrw^{ty}_{\beta\teZprolp})^{*} T' (\leftrightarrow_{pi}^{ty})^{*}
      U' (\rlZwr^{ty}_{\beta\teZprolp})^{*} U\).
      Because \(\rlZrw_{\beta\teZprolp}\) preserves typing
      (\cref{prop:pcert-eq-confluence}),
      we have \(\Gamma \jgZpvs U': s\).
      By~\cite[Lemma 43]{blanqui01phd}, \(\Gamma \jgZpvs T: s\).
      By \cref{prop:preservation-equivalence},
      \(\trZobj{T}_{\Gamma} \lpZdeq \trZobj{U}_{\Gamma}\)

      If \(s = \teZProp\), then
      \(\trZtyp{T}_{\Gamma} = \ecZPrf \trZobj{T}_{\Gamma} \lpZdeq
      \ecZPrf \trZobj{U}_{\Gamma} = \trZtyp{U}_{\Gamma}\).
      Moreover we have \(\trZctx{\Gamma} \jgZlpme \trZtyp{T}_{\Gamma}:
      \lpZTYPE\) because, by induction hypothesis,
      \(\trZobj{T}_{\Gamma} : \trZtyp{\teZProp} = \ecZEl \trZobj{\teZProp} =
      \ecZEl \ecZprop = \ecZProp\), and
      \((p: \ecZProp \vdash_{\sigPcertEnc} \ecZPrf p: \lpZTYPE: \lpZKIND)\).
      If \(s = \teZType\),
      \(\trZtyp{T}_{\Gamma} = \ecZEl \trZobj{T}_{\Gamma} \lpZdeq
      \ecZEl \trZobj{U}_{\Gamma} = \trZtyp{U}_{\Gamma}\).
      By induction hypothesis, \(\trZobj{T}_{\Gamma}: \trZtyp{\teZType}_{\Gamma}
      = \ecZType\).
      If \(s = \teZKind\), then \(T = U = \teZType\)
      (\(\teZType\) is the only inhabitant of \(\teZKind\)).
      Finally, \(\trZtyp{\teZType} = \ecZType: \lpZTYPE\).

    \item[sig]
      \(\prftree[r]{$\Sigma(f) = (\overrightarrow{x, T}, U, s)$}{%
      \overrightarrow{x: T} \vdash U: s}{%
      \left(\Gamma \vdash t_i: \teZsubsp{T_i}{\left(x_j \mapsto t_j\right)_{j < i}}\right)_i
      }{\Gamma \vdash f(\vec{t}): \teZsubsp{U}{\overrightarrow{x \mapsto t}}}\)

      We first observe from \cref{fig:pcert-encoding}
      that for each \(f \in \Sigma^{\text{\pvs}}\), we have a counterpart
      symbol \(\hat{f} \in \sigPcertEnc\)
      such that if \(\Sigma^{\text{\pvs}}(f) =
      \left(\overrightarrow{x: T}, U, s\right)\),
      then
      \(\sigPcertEnc(\hat{f}) =
      \left(\overrightarrow{x, \trZtyp{T}},
      \trZtyp{U}_{\overrightarrow{x: T}}, \lpZTYPE\right)\).

      By induction hypothesis, for each \(i\), we have
      \(\trZctx{\Gamma} \jgZlpme \trZobj{t_{i}}_{\Gamma}:
      \trZtyp{\teZsubsp{T_{i}}{(x_{j} \mapsto t_{j})_{j < i}}}_{\Gamma}\)
      which we can write as, thanks to \cref{prop:preservation-substitution},
      \(\trZctx{\Gamma} \jgZlpme \trZobj{t_{i}}_{\Gamma}:
      \teZsubsp{\trZtyp{T_{i}}_{\Gamma}}{(x_{j} \mapsto
      \trZobj{t_{j}}_{\Gamma})_{j < i}}\).

      Now, using the signature rule, we are able to conclude
      \(\trZctx{\Gamma} \jgZlpme \hat{f}\,\overrightarrow{\trZobj{t}_{\Gamma}}:
      \teZsubsp{\trZtyp{U}}{\overrightarrow{x \mapsto \trZobj{t}}}\).
      By \cref{prop:preservation-substitution}, we obtain
      \(\trZctx{\Gamma} \jgZlpme
      \hat{f}\,\overrightarrow{\trZobj{t}_{\Gamma}}:
      \trZtyp{\teZsubsp{U}{\overrightarrow{x \mapsto t}}}\).
      Moreover, we have taken care to define the translation
      in \cref{fig:pcert-to-lpme} such that
      \(\trZobj{f(\overrightarrow{t})} = \hat{f}\,\overrightarrow{\trZobj{t}}\).
      \qedhere
  \end{description}
\end{proof}


%% file: eq_to_rewriting.tex
\section{Mechanised Type Checking}

The encoding of \pcert{} into \lpme{} can be used to proof check terms of
\pcert{} using a type checker for \lpme{}. But because of the rule
\begin{equation}
  \label{eq:lpme-conv-rule}
  \prftree{\Gamma \vdash t: B}{\Gamma \vdash A: s}{A \rlZdeq B}
  {\Gamma \vdash t: A}\tag{\lpme{}-conv}
\end{equation}
type checking is decidable only if \(\rlZdeq\) is.
A decidable relation equivalent to \(\rlZdeq\) can be obtained using
the convertibility relation stemming from the rewriting relation
of a convergent rewrite system,
yielding the type system \lpmr (\(R\) for \emph{rewriting}).
Consequently, while type checkers cannot be provided for \lpme in general,
they can for \lpmr, as can be seen with
\Dedukti\footnote{\url{https://github.com/Deducteam/lambdapi.git}}.
Such rewrite systems can be obtained through
\emph{completion procedures}~\cite{baader90trat}.
However, completion procedures rely on a well-founded order
that cannot be provided here because of \cref{eq:pcert-eq-pair-pi}
which cannot be oriented since each side of the equation has a free variable
which is not in the other side.

A possible solution would be to rewrite all proofs of a pair to a
canonical proof with a rule of the form
{
  \newcommand\canon{\operatorname{\enc{canon}}}
\begin{equation*}
  \ecZpair t\, p\, m\, h \rlZred \ecZpair t\, p\, m\, (\canon t\, p\, m)
\end{equation*}
where
\(t: \ecZType, p: \ecZEl t \teZarr \ecZProp, m: \ecZEl t \vdash \canon t\,p\,m:
\ecZPrf (p\,m): \lpZTYPE\).
But this creates a rewrite rule that duplicates three variables.
}

Otherwise, as noted in~\cite{Knuth1983SimpleWP}, the addition of a symbol to the
signature can circumvent the issue.
Hence, we add a symbol for proof irrelevant pairs, and make it equal to pairs
\begin{gather}
  \label{sym:pairp}
  t: \ecZType, p: \ecZEl t \teZarr \ecZProp, m: \ecZEl t
  \vdash \ecZpairp t\, p\, m: \ecZEl (\ecZpsub t\, p): \lpZTYPE
  \\
  \label{eq:eq-pair-pairp}
  \ecZpair t\, p\, m\, h = \ecZpairp t\, p\, m
\end{gather}
thus
\((\ecZpair\ t\ p\ m\ h) \rlZdeq (\ecZpairp\ t\ p\ m) \rlZdeq
(\ecZpair\ t\ p\ m\ h')\).
The new set of identities given by
\cref{eq:eq-lhol-bool,eq:eq-lhol-forall,eq:eq-lhol-impd,eq:eq-lhol-arrd,%
eq:pcert-eq-prol,eq:eq-pair-pairp}
can be completed into a rewrite system \(R\) which is equivalent
to the equations:

\begin{figure}[ht]
  \begin{minipage}{0.5\textwidth}
    \begin{gather}
      \label{eq:rw-beta}
      (\teZabst{x}[T]{t})\, u \rlZrw \teZsubs{t}{x}{u}\\
      \label{eq:rw-pair-pairp}
      \ecZpair t\, p\, m\, h \rlZred \ecZpairp t\, p\, m\\
      \label{eq:rw-fst}
      \ecZfst t_{0}\, p_{0}\, (\ecZpairp t_{1}\, p_{1}\, m) \rlZred m
    \end{gather}
  \end{minipage}
  \begin{minipage}{0.5\textwidth}
    \begin{gather}
      \label{eq:rw-lhol-bool}
      \ecZEl \ecZprop \rlZred \ecZProp\\
      \label{eq:rw-lhol-forall}
      \ecZPrf (\ecZfa t ~ p) \rlZred \teZprod{x}{\ecZEl t}{\ecZPrf(p~x)}\\
      \label{eq:rw-lhol-arrd}
      \ecZEl (t \ecZarrd u) \rlZred \teZprod{x}{\ecZEl t}{\ecZEl (u ~ x)}\\
      \label{eq:rw-lhol-impd}
      \ecZPrf (p \ecZimpd q) \rlZred \teZprod{h}{\ecZPrf p}{(\ecZPrf (q ~ h))}
    \end{gather}
  \end{minipage}
  \caption{Rewrite system \(R\) resulting from the completion
    of the equations of the encoding of \pcert in
    \lpme.}\label{fig:lpmr-rewrite-system}
\end{figure}

\begin{proposition}\label{prop:equivalence-lpme-lpmr}
  Let \(\rlZrw_{R}\) be the closure by context and substitution of the rewrite
  rules of \cref{fig:lpmr-rewrite-system}, and
  \(\rlZcov_{R}\) be the smallest equivalence containing
  \(\rlZrw_{R}\).
  Then, for all \(M, N \in {\mathcal{T}(\sigPcertEnc,
    \mathcal{S}^{\mathrm{\lambda\Pi}}, \mathcal{V})}\), if
  \(M \lpZdeq N\) then \(M \rlZcov_{R} N\).
\end{proposition}

\begin{proof}
  It suffices to prove that every equation of \pcert is included
  in $\rlZcov_{R}$. This is immediate for the Equations
  \labelcref{eq:beta,eq:eq-lhol-bool,eq:eq-lhol-forall,eq:eq-lhol-impd,eq:eq-lhol-arrd}
  since they are equal to the rules
  \labelcref{eq:rw-beta,eq:rw-lhol-bool,eq:rw-lhol-forall,eq:rw-lhol-impd,eq:rw-lhol-arrd}. For
  the Equation \labelcref{eq:pcert-eq-pair-pi}, we have $\ecZpair t\,
  p\, m\, h_0\hookrightarrow_R \ecZpairp t\, p\, m\hookleftarrow_R
  \ecZpair t\, p\, m\, h_1$. Finally, for the
  Equation \labelcref{eq:pcert-eq-prol}, we have $\ecZfst t_{0}\, p_{0}\,
  (\ecZpair t_{1}\, p_{1}\, m\, h)\hookrightarrow_R \ecZfst t_{0}\,
  p_{0}\, (\ecZpairp t_{1}\, p_{1}\, m) \hookrightarrow_R m$.\qedhere
\end{proof}

\begin{remark}
  \begin{itemize}
    \item Rewrite system \(R\) is confluent because it is orthogonal.
    \item Termination of \(R\) is required to obtain the decidability of
    \(\rlZcov_{R}\). A possible approach to prove it would be to extend
    the termination model of \lhol described in~\cite{dowek17icalp}.
    \item In order to prove the completeness of the encoding, that is, the
      fact that a type is inhabited whenever its encoding is, it
      could be useful to have the reciprocal implication, that is, if
      \(M \rlZcov_{R} N\) and \(M, N \in \mathcal{T}(\sigPcertEnc,
      \mathcal{S}^{\mathrm{\lambda\Pi}}, \mathcal{V})\), then \(M
      \lpZdeq N\). We leave this for future work too.
  \end{itemize}
\end{remark}

A priori, the introduction of $\ecZpairp$ allows one to craft terms
that cannot be proof checked in \pcert.
Indeed, given a predicate \(\operatorname{Even}\) on natural numbers,
the term \((\ecZpairp \N\, \operatorname{Even}\, 3)\)
is the encoding of
\((\teZpair{\N}{\operatorname{Even}}{3}{h})\)
which cannot be type checked in \pcert since there is no proof \(h\) that 3 is even.
However, \Dedukti relies on a system of modules and tags attached to symbols
to define where and how symbols can be used.
A symbol tagged \emph{protected} cannot be used to build terms outside
of the module where it is defined,
but it may appear during type checking because of conversion, a trick first introduced in \cite{thire19types} and used also for encoding Cumulative Type Systems in \lpme \cite{thire20phd}.
In our case, one may protect \(\ecZpairp\) in the module that defines
the encoding of \pcert, so that users of the encoding are forced to use
\(\ecZpair\).


%% file: conclusion.tex
\section*{Conclusion}

This work provides an encoding of predicate subtyping with proof
irrelevance into the \(\mathrm{\lambda\Pi}\)-calculus modulo theory,
\lpme \cite{assaf19draft}. We first recall \pcert, an extension of
higher-order logic with predicate subtyping and proof irrelevance
\cite{gilbert2018phd}. We then provide a \lpme signature to encode
terms of \pcert, and prove that the encoding is correct: if a \pcert
type is inhabited, then its translation in \lpme is inhabited too.
Finally, we show that the equational theory of our encoding is
equivalent to a confluent set of rewrite rules which enable us to use
\Dedukti to type check encoded specifications.

However, two important problems are left open. First, is our encoding
complete, that is, is a \pcert type inhabited if its translation is?
Second, is the confluent rewrite system used in the encoding
terminating? We believe that these two properties hold but leave their
difficult study for future work.

\paragraph*{Perspectives}

The encoding of \pcert in \lpmr is the stepping stone towards an automatic
translator from \pvs to \Dedukti.
Indeed, \pvs does not have proof terms in its syntax, and consequently
type checking is undecidable. The creation of \pcert allows to convert
\pvs terms to a syntax whose type checking is decidable.
This was the work of F.~Gilbert in~\cite{gilbert2018phd}.
Now we are able to express this decidable syntax in \lpmr and hence in \Dedukti.
However, the type system proposed here only allows to coerce from
a type to its direct supertype or a subtype, that is,
we can go from \((\teZpsub{(\teZpsub{\iota}{P})}{Q})\) to
\(\teZpsub{\iota}{P}\) in one coercion, but we cannot coerce from
\((\teZpsub{(\teZpsub{\iota}{P})}{Q})\) to \(\iota\),
whereas \pvs can.
Consequently, an algorithm to elaborate the correct sequence of coercions
is needed to obtain terms that can be type checked in \Dedukti.

Other features of \pvs can be integrated into \pcert and the
encoding: dependent types like \((\teZpsub{\,list}{(\teZabst{\ell}{\operatorname{length}\,\ell = n})})\),
recursive definitions of functions, and dependent records.
With those features encoded, almost all
the standard
library\footnote{\url{ http://www.cs.rug.nl/~grl/ar06/prelude.html }} of
\pvs can be translated to \Dedukti.

Finally, while the previous points were concerned with the translation of
specifications from \pvs,
we may also want to translate proofs developed in \pvs.
These proofs are witnesses of \emph{type correctness conditions} (\tcc), which are required to type check terms.
Since \pvs is a highly automated prover, proof terms often come from application
of complex tactics that cannot be mimicked into \Dedukti.
However, proof terms may either be provided by hand, emulating the interaction
provided by \tcc's, or we may call external solvers
\cite{ElHaddad19Ekstrakto}.

\subparagraph*{Acknowledgements} The authors thank Gilles Dowek and the
anonymous referees very much for their remarks.
